\documentclass{amsart}

\usepackage{graphicx}          

\usepackage{graphics} 
\usepackage{epsfig} 
\usepackage{mathptmx} 
\usepackage{times} 
\usepackage{amsfonts}
\usepackage{amsmath} 
\usepackage{amssymb}  
\usepackage{arydshln}
\usepackage{mathtools}

\usepackage[outdir=./]{epstopdf}
\usepackage{pgfplots}
\usetikzlibrary{positioning}
\usepackage{datatool,tikz,filecontents}
\pgfplotsset{compat=1.14}
\usepackage[colorinlistoftodos, color=blue!20]{todonotes}
\usepackage{hyperref}
\definecolor{skyblue}{rgb}{0,0.4470,0.7410}
\definecolor{redfaded}{rgb}{0.6350,0.0780,0.1840}
\definecolor{lightgreen}{rgb}{0.5280,0.7440,0.3640}
\definecolor{newgreen}{rgb}{0.2480,0.6000,0.3240}
\newtheorem{remark}{Remark}
\usepackage{enumitem}
\newcommand{\secref}[1]{\S\ref{#1}} 

\newcommand{\norm}[1]{\left\|#1\right\|}
\newcommand{\op}[1]{\mathring{#1}}
\newcommand{\R}{\mathbb{R}}
\newcommand{\Let}{\coloneqq}
\DeclareMathOperator{\Ad}{Ad}
\DeclareMathOperator{\e}{exp}
\DeclareMathOperator*{\minimize}{minimize}

\newtheorem{theorem}{Theorem}

\newenvironment{customlegend}[1][]{%
    \begingroup
    \csname pgfplots@init@cleared@structures\endcsname
    \pgfplotsset{#1}%
}{%
    \csname pgfplots@createlegend\endcsname
    \endgroup
}%

\def\addlegendimage{\csname pgfplots@addlegendimage\endcsname}

\pgfkeys{/pgfplots/number in legend/.style={%
        /pgfplots/legend image code/.code={%
            \node at (0.295,-0.0225){#1};
        },%
    },
}

\begin{document}

\title{Exact isoholonomic motion of the planar Purcell's swimmer}


\author{Sudin Kadam}
\address{Systems and Control Engineering Department, Indian Institute of Technology Bombay, Mumbai, India}
\curraddr{}
\email{sudin@sc.iitb.ac.in}
\thanks{}

\author{Karmvir Singh Phogat}
\address{Systems and Control Engineering Department, Indian Institute of Technology Bombay, Mumbai, India}
\curraddr{Electrical Engineering, Korea Advanced Institute of Science and Technology, Daejeon, South Korea}
\email{karmvir.p@sc.iitb.ac.in}
\thanks{}

\author{Ravi N. Banavar}
\address{Systems and Control Engineering Department, Indian Institute of Technology Bombay, Mumbai, India}
\curraddr{}
\email{banavar@iitb.ac.in}
\thanks{}

\author{Debasish Chatterjee}
\address{Systems and Control Engineering Department, Indian Institute of Technology Bombay, Mumbai, India}
\curraddr{}
\email{dchatter@iitb.ac.in}
\thanks{}

\date{}

\dedicatory{}

\begin{abstract}                          
In this article we present the discrete-time isoholonomic problem of the planar Purcell's swimmer and solve it using the Discrete-time Pontryagin maximum principle. The 3-link Purcell's swimmer is a locomotion system moving in a low Reynolds number environment. The kinematics of the system evolves on a principal fiber bundle. A structure preserving discrete-time kinematic model of the system is obtained in terms of the local form of a discrete connection. An adapted version of the Discrete Maximum Principle on matrix Lie groups is then employed to come up with the necessary optimality conditions for an optimal state transfer while minimizing the control effort. These necessary conditions appear as a two-point boundary value problem and are solved using a numerical technique. Results from numerical experiments are presented to illustrate the algorithm.
\end{abstract}
\keywords{low Reynolds number swimming; discrete optimal control; principal fiber bundle; Purcell's swimmer.}
\maketitle

\section{INTRODUCTION}\label{sec:Intro}
In the microscopic world of biological systems where inertia is negligible, viscosity dominates motion. This effect is observed in
the very low Reynolds numbers regime, in which the viscous forces far dominate the inertial forces in such biological systems.  This effectively implies that the mechanism comes to a complete halt as soon as the propulsion through changes in the shape of the biological system is stopped. A vast majority of living organisms are found to perform motion at microscopic scales at such low Reynolds number conditions.
From the modelling and mathematical perspective, for a large class of such locomotion systems the configuration space is amenable to the framework of a
principal fiber bundle \cite{ostrowski1996geometric}, \cite{bloch1996nonholonomic}, and geometric tools lead to insightful solutions to control theoretic problems. 
The configuration variables are partitioned as the {\it base} and the {\it group} variables, and the former are, usually, fully actuated.
The low Reynolds number effect gives rise to a principal kinematic form of the equations of motion, and further, the shape space as well as the structure group of many of these systems is either the Special Euclidean group $SE(3)$ or one of its subgroups; the reader is referred to \cite{kelly1995geometric}, \cite{ostrowski1998geometric} for many illustrative examples. 
%

A gait of a locomotion system on a principal fiber bundle corresponds to a closed curve in the base space of the bundle. Optimal gait design, i.e., the design of gaits that minimize control energy, is one of the open problems in the control of locomotion systems. The isoholonomic problem falls in this class: it pertains to finding closed loops in the base space that result in the desired displacement in the structure group (the gross motion of the body) with minimal control effort (movement of the limbs). The falling cat problem is one of the famous examples in this context \cite{montgomery1990isoholonomic}, and other examples include \cite{krishnaprasad1991control}, \cite{koon1997reduction}, \cite{cabrera2008base}. In \cite{krishnaprasad1991control}, for instance, the authors have solved an optimal control problem for a similar form of principal kinematic system, but in the continuous time setting. We note that a closed form solution was obtained in that case because of the amenable form of system equations. In this article we present results for the isoholonomic problem for the  3-link planar Purcell's swimmer \cite{purcell1977life}, the simplest possible microswimmer. This mechanism has given rise to considerable research in the areas of modelling, control, optimal gait design, etc.; see \cite{passov2012dynamics}, \cite{melli2006motion}, \cite{tam2007optimal}, and the references therein. However, to the best of our  knowledge, the discrete-time isoholonomic problem has not been addressed in the literature so far.


A few well-established approaches to solve the continuous-time optimal control problem involve variational principles or, more generally,  applications of the Pontryagin Maximum Principle (PMP), and/or dynamic programming techniques \cite{agrachev2013control}, \cite{milyutin1998calculus}, \cite{pontryagin1987mathematical}. However, state and control inequality constraints are not easily handled with these techniques. The necessary mathematical conditions of optimality obtained via either of the first two approaches invariably tend to be difficult to solve analytically for most engineering systems, and the last technique, further, suffers from the curse of dimensionality. More importantly, however, we would like to obtain a time-discretization of the control law to enable implementation. A conventional approach to discretization proceeds along one of the following two ways: In the first, the entire control synthesis is performed in the continuous time and the final controller is then discretized for the purpose of implementation. This is, therefore, an approximation and the performance deviates from those predicted by the continuous-time models, in particular, for nonlinear systems evolving on non-flat manifolds. In the second, the continuous-time plant is initially discretized, and then the synthesis is carried out in the discrete domain. Here, the discretization is not exact and does not preserve inherent mechanical system invariants like momentum and/or energy and/or constraints. Adopting the discrete mechanics approach, as we do here, preserves such invariants of the system and ensures that the discretization is exact \cite{kobilarov2010geometric}.

A general discrete-time PMP  that addresses a wide class of discrete-time optimal control problems in Euclidean spaces with state and action constraints appears in the work of Boltyanski \cite{bolt1975method}. A discrete-time PMP for systems evolving on matrix Lie groups was established in \cite{KarmAutomatica} by an extension of Boltyanski's techniques. In this article, with the discrete mechanics based model of the Purcell's swimmer, we adopt the approach outlined in \cite{KarmAutomatica} to obtain a discrete time control law. We arrive at necessary optimality conditions by extension of the results on matrix Lie groups to the principal kinematic form of systems where the base manifold evolves on $\mathbb{R}^n$. The discrete-time optimal control, solved by applying the PMP on matrix Lie groups, leads to a two point boundary value problem, which is then solved numerically to obtain the optimal control sequences. The utility of this work is that the results obtained can be extended to wide class of systems satisfying the principal kinematic form of equations \cite{ostrowski1996geometric}, \cite{bloch1996nonholonomic}, \cite{krishnaprasad1991control}.

This article is organized as follows. In \secref{sec:prelims} we present the mathematical preliminaries and define the topology of locomotion systems along with the principal kinematic form of equations. \secref{sec:isoholonomic_problem} explains the continuous time isoholonomic problem. We then introduce the model of the planar Purcell's swimmer in the principal kinematic form. A structure-preserving discretization of its kinematic equations in terms of the discrete connection form is presented in this section. \secref{sec:DMP} presents the formulation of the discrete isoholonomic problem for the swimmer model, followed by the necessary optimality conditions obtained by using the discrete maximum principle on matrix Lie groups. In \secref{sec:simulation}, the results of numerical experiments are presented.
\section{Preliminaries}\label{sec:prelims}
We provide a brief overview of basic notations and definitions from differential geometry that  are frequently encountered in this article. 

Let $G$ be a Lie group with $e$ as the group identity, and $\mathfrak{g}$ be the Lie algebra of the Lie group $G$. Let $ G \times G  \ni (g,h) \mapsto \Phi_g (h) \Let gh \in G$ be the \textit{left action} on the Lie group $G$  by itself. We further assume that the left group action $\Phi_{g}$ is \textit{free} , i.e., for all $g, h \in G$, $\Phi_{g}(h)=h$ implies $g=e$. 
The \textit{tangent lift} of $\Phi, T\Phi : G \times TG \to TG$ is the action
\begin{align}
T\Phi\big(g,(h,v)\big) \Let \big(\Phi_{g}(h),T_{h}\Phi_{g}(v)\big),
\end{align}
where $T_{h}\Phi_{g}(v)\Let \left.\frac{d}{ds}\right|_{s=0}\Phi_{g}\big(\gamma(s)\big)$ with a smooth curve $\gamma$ on $G$ such that $\gamma(0)=h$ and $\dot{\gamma}=v.$
The \textit{cotangent lift} of $\Phi$, $T^*\Phi : G \times T^*G \to T^*G$ is the action
\begin{equation}
T^*\Phi \big(g,(w,a)\big) \Let \big(\Phi_{g}(w), T_{\Phi_{g}(w)}^*\Phi_{g^{-1}}(a)\big),
\end{equation}
where $\big\langle T_{\Phi_{g}(w)}^*\Phi_{g^{-1}}(a), v \big\rangle \Let \big\langle a , T_{\Phi_{g}(w)}\Phi_{g^{-1}}(v)\big\rangle $ for all $a \in T^*_{w}G$, $v \in T_{\Phi_{g}(w)}G$. The \textit{adjoint action} of $G$ on $\mathfrak{g}$ is defined as 
\begin{equation}
G \times \mathfrak{g} \ni (g, \beta) \mapsto \Ad_g\beta := \left. {\frac{d}{ds}}\right|_{s=0} g\e^{s\beta}g^{-1} \in \mathfrak{g}.
\end{equation}
The \textit{coadjoint action} of $G$ on $\mathfrak{g}^*$ is the inverse dual of the adjoint action, given by 
\begin{equation}\label{coAdjoint_action}
G \times \mathfrak{g}^* \ni (g,a) \mapsto Ad^*_{g^{-1}}a \in \mathfrak{g}^*, 
\end{equation}
where $\langle \Ad_{g^{-1}}^*a, \beta \rangle \Let \langle a, \Ad_{g^{-1}} \beta \rangle$ for all $\beta \in \mathfrak{g}$ and $a \in \mathfrak{g}^*$
Consider a trivial principal fiber bundle $Q$ with the base manifold $M$ and the structure group $G$ as $Q \Let M \times G.$ A free left action of $G$ on $Q$ is induced by the left action of $G$ on the group component of $Q$ as  
    \[
   Q \ni (x,g) \mapsto \phi_h(x,g) \Let (x,hg) \in Q \;\text{ for all }\; h \in G,
    \]
see Figure \ref{Principle_fiber_bundle}.
\begin{figure}[!htb]
\centering
\includegraphics[scale=.53]{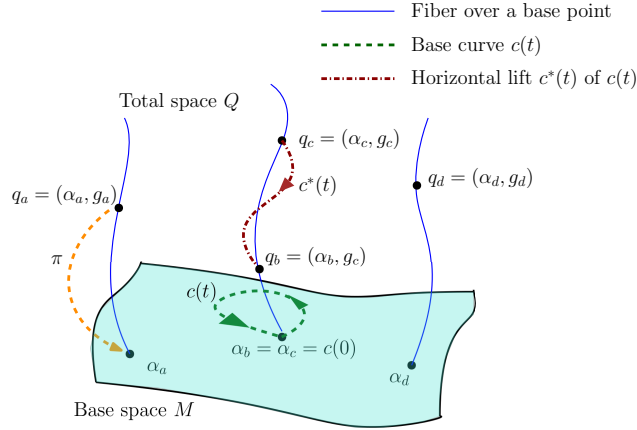}
\caption{Principal fiber bundle, horizontal lift and holonomy}
\label{Principle_fiber_bundle}
\end{figure}
\section{The planar Purcell's swimmer and the isoholonomic problem}\label{sec:isoholonomic_problem}
The planar Purcell's swimmer is a micro-swimmer made up of $3$ slender links with two rotary joints; see Figure \ref{Purcell_swimmer}. The configuration manifold $Q \Let SO(2) \times SO(2) \times SE(2)$ of the Purcell's swimmer admits the topology of a trivial principal fiber bundle $Q = M \times G $, where the base manifold $M\Let SO(2) \times SO(2)$ describes the configuration of the internal shape variables of the mechanism, and  the Lie group $G\Let SE(2)$ describes the position and the orientation of the micro-swimmer in $X-Y$ plane. 
\begin{figure}[!htb]
\centering
\includegraphics[scale=.36]{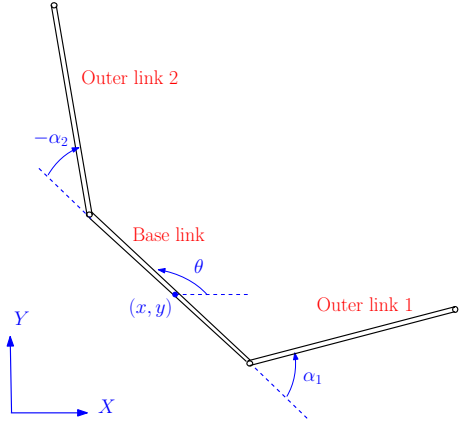}
\caption{The Purcell's swimmer}
\label{Purcell_swimmer}
\end{figure}
We further consider that the \textit{Outer Link} 1 and the \textit{Outer Link} 2 are not allowed to overlap with \textit{Link} 0; that is ensured if $\alpha_1, \alpha_2 \in \, ]-\pi,\pi[.$ Based on the Cox theory,  the kinematic model of the Purcell's swimmer depends on the lengths of the three links and viscous drag coefficient \cite{hatton2013geometric}, and is given by
\begin{equation}\label{kinematic}
\begin{aligned}
\dot{\alpha}(t) &=u(t),  \\
\dot{g}(t) & = - g(t) A\big(\alpha(t)\big) \dot{\alpha}(t), 
\end{aligned}
\end{equation}
where $\R \ni t\mapsto\big(\alpha(t),g(t)\big) \in M \times G$ defines the state trajectory of the system, $\R \ni t \mapsto u(t) \in \R^2$ defines the control trajectory of the system and  $TM \ni (a,v) \mapsto A(a)v \in \mathfrak{g}$ is the Lie algebra valued local connection form. 

\subsection{Isoholonomic problem}
Before we formally define the isoholonomic problem, let us introduce the notions related to holonomy for the system \eqref{kinematic}. For a principal kinematic system with trajectory $\R \ni t \mapsto q(t)=\big(\alpha(t),\, g(t)\big) \in M \times G$, the tangent space $T_{q(t)}Q$ for each $t$, can be split into the horizontal subspace $H_{q(t)}Q$ and the vertical subspace $V_{q(t) }Q$ using the system's connection form\footnote{A detailed discussion on connection forms and kinematic systems on trivial principal fiber bundles may be found in \cite{bloch1996nonholonomic}, \cite{kelly1995geometric}}. The horizontal subspace $H_{q(t)}Q$ for the system \eqref{kinematic} is defined through the local connection form as 
\[
H_{q(t)}Q = \big\{\big(\dot{\alpha}(t), - g(t) A\big(\alpha(t)\big) \dot{\alpha}(t)\big) \in T_{q(t)}Q\big\}.
\]
The horizontal lift of a continuous curve $\gamma : [0,T] \to M$ passing through a point $\bar{\alpha} \in M$ is a continuous curve $\gamma^* : [0,T] \to Q$ such that the following hold:
\begin{enumerate}
\item $\pi \big(\gamma^*(t)\big) = \gamma(t) \:\: for\:\: t \in [0,T]$,
\item $\frac{d}{dt} \gamma^*(t) \in H_{\gamma(t)}Q \:\: for\:\: t \in [0,T].$
\end{enumerate}
The geometric phase or holonomy of a continuous curve $\gamma : [0,T] \to M$ with $\gamma(0) = \gamma(T) = \bar{\alpha} \in M$ is $\bar{g} \in G$ that is determined by the horizontal lift $\gamma^*$ of $\gamma$ as
\[
\bar{g} \Let  g_0^{-1}g_T, \quad \text{where} \; \gamma^*(0) = (\bar{\alpha},g_0) \text{ and } \gamma^*(T) = (\bar{\alpha},g_T),
\]
see Figure \ref{Principle_fiber_bundle}.

{\it Isoholonomic problem \cite{montgomery1990isoholonomic}, \cite{koon1997reduction}}:  Among all loops with a given holonomy find a loop in the base space $M$ that corresponds to the least control effort. To define it formally, for a given holonomy $\bar{g} \in G$ and a fixed point $\bar{\alpha} \in M$, the isoholonomic problem for the locomotion kinematics \eqref{kinematic} is defined as 
\begin{equation}
\label{eq:sopt}
\begin{aligned}
\minimize_{u} &&& \mathcal{J}(u) \Let \int_0^T c(u(\tau)) d \tau\\
\text{subject to} &&&
\begin{cases}
\text{System kinematics } \eqref{kinematic},\\
\alpha(0)= \alpha(T)=\bar{\alpha}, \; g(0)^{-1}g(T)=\bar{g},\\
\bar{\alpha} \; \text{ and } \; \bar{g} \; \text{ are fixed},
\end{cases} 
\end{aligned} 
\end{equation}
where $\R^2 \ni b \mapsto c(b) \geq 0 \in \R$ accounts for the running cost. 

	The continuous-time optimal control problem \eqref{eq:sopt} is typically solved numerically, and that requires discretization of the system kinematics \eqref{kinematic}. Therefore, we turn to derive a variational integrator of the system kinematics \eqref{kinematic} that preserves the manifold structure and the system symmetries unlike standard discretization schemes such as Euler's steps and its derivatives \cite{marsden2001discrete}. Further, we define the optimal control problem \eqref{eq:sopt} in discrete-time and solve it using discrete-time Pontryagin's maximum principle on matrix Lie groups \cite{KarmAutomatica}.  
\subsection{Discrete-time kinematic model of the Purcell's swimmer}
We now derive a variational integrator of the Purcell's swimmer \eqref{kinematic} using the discrete mechanics approach
described in \cite{kobilarov2010geometric}.
Suppose $[N]$ denote the integers from zero to $N \in \mathbb{N}$. Let us discretize the time horizon $[0,\; T]$ uniformly in $N$ subintervals such that the system configurations and the control actions at discrete-time instances $\{t_k \Let k h\; |\; k \in [N]\}$ with a fixed step length $h\Let T/N>0$ are given by
\[g_k \Let g(t_k), \quad \alpha_k \Let \alpha(t_k),\quad u_k \Let u(t_k) \quad \text{for } k \in [N].\]
We assume that the system kinematics \eqref{kinematic} is actuated by a piecewise constant control, i.e., 
\begin{equation}\label{eq:app_control}
[0, T] \ni t \mapsto u(t) \Let u_k \in T^*M \text{ for } t \in [t_k, t_{k+1}{[},
\end{equation}
and the local connection form in \eqref{kinematic} is approximated as
\begin{equation}\label{eq:app_connection}
[0, T] \ni t \mapsto A\big(\alpha(t)\big) \Let A(\alpha_k) \in T^*M \text{ for } t \in [t_k, t_{k+1}{[}.
\end{equation}
Under the approximations \eqref{eq:app_control} and \eqref{eq:app_connection}, the discrete-time kinematics system for the Purcell's swimmer is derived by integrating the continuous-time kinematics \eqref{kinematic} as
\begin{equation}\label{dis_kin}
\begin{aligned}
\alpha_{k+1} &= \alpha_k + h u_k, \\
g_{k+1} &= g_k \e \big(-h A(\alpha_k) u_k \big)  
\end{aligned}
\end{equation}
where $\e: \mathfrak{g} \rightarrow G$ is the exponential map.\footnote{The exponential of $X \in \mathfrak{g}$ is the map $ \mathfrak{g} \ni X \mapsto \e(X) = \gamma(1) \in G,$ where $\gamma : \mathbb{R} \to G$ is the one-parameter subgroup generated by $X$, or equivalently the integral curve of $X$ starting at the group identity \cite[p.\ 522]{lee2003smooth}.}
\begin{remark}
Note that the preceding discretization scheme has the following properties \cite{kobilarov2010geometric}:
\begin{itemize}
    \item It preserves the configuration manifold of the system; i.e., the states always remain on the original configuration manifold, and
    \item The horizontal subspace $H_qQ$ at all $q \in Q$ is equivariant with respect to the group action.
\end{itemize}
\end{remark}

\section{Discrete isoholonomic problem for the planar Purcell's swimmer}\label{sec:DMP}
Let us define the discrete-time isoholonomic problem for the variational integrator \eqref{dis_kin} and then derive the necessary conditions for optimality for the discrete-time isoholonomic problem using the discrete-time Pontryagin's maximum principle \cite{KarmAutomatica} to arrive at an optimal control function. 
\subsection{Isoholonomic problem in discrete-time}
Considering a quadratic cost function 
\[
\R^2 \ni \mu \mapsto c(\mu) \Let \frac{1}{2}\norm{\mu}_2^2 \in \R, 
\]
the isoholonomic problem \eqref{eq:sopt} is defined in discrete-time as 
\begin{equation}
\label{eq:dsopt}
\begin{aligned}
\minimize_{\{u_k\}_{k=0}^{N-1}} &&& J(u) \Let \sum_{k=0}^{N-1}\frac{h}{2}\norm{u_k}_2^2\\
\text{subject to} &&&
\begin{cases}
\text{System kinematics } \eqref{dis_kin},\\
\alpha_0 = \alpha_N =\bar{\alpha}, \:\: g_0^{-1}g_N=\bar{g},\\
\bar{\alpha} \; \text{ and } \; \bar{g} \; \text{ are fixed}.
\end{cases} 
\end{aligned} 
\end{equation}
It is important to note that the problem \eqref{eq:dsopt} is not a fixed endpoint problem  because of the holonomy constraint, i.e., $ g_0^{-1}g_N=\bar{g} $. However, using the group invariance property of the system kinematics \eqref{dis_kin} and the cost functions $c$, the discrete-time optimal control problem \eqref{eq:dsopt} is translated to a fixed endpoint problem as explained in \secref{ssec:transfixedopt}.

\subsection{Translation of \texorpdfstring{\eqref{eq:dsopt}{} }tto a fixed endpoint problem} \label{ssec:transfixedopt}
First we prove that the optimal control problem \eqref{eq:dsopt} is invariant under the left translation of the initial group configuration $g_0 \in G$, or equivalently the optimal control problem is insensitive to the initial group configuration of the system.
Considering $g_0=\bar{g}_0 \in G$, the final group configuration of the system $g_N$ is given by 
\begin{equation*}
g_N = \bar{g}_0 \e \big(-h A(\alpha_0) u_0 \big) \circ\: \dots \: \circ \e \big(-h A(\alpha_{N-1}) u_{N-1}\big) 
\end{equation*}
that leads to the following holonomy condition 
\[
g_0^{-1}g_N = \e \big(-h A(\alpha_0) u_0 \big) \circ \dots \circ \e \big(-h A(\alpha_{N-1}) u_{N-1} \big)
\]
which is independent of the initial group configuration $\bar{g}_0$. In particular, choosing the initial group configuration as the group identity, i.e., $\bar{g}_0=e$, the discrete isoholonomic problem is defined as a fixed end point problem as: 
\begin{equation}
\label{eq:dsopt_fixedendpoint}
\begin{aligned}
\minimize_{\{u_k\}_{k=0}^{N-1}} &&& J(u) \Let \sum_{k=0}^{N-1}\frac{h}{2}\norm{u_k}_2^2\\
\text{subject to} &&&
\begin{cases}
\text{System kinematics } \eqref{dis_kin},\\
\alpha_0 = \alpha_N =\bar{\alpha}, \:\: g_0= e, \:\: g_N=\bar{g},\\
\bar{\alpha} \; \text{ and } \; \bar{g} \; \text{ are fixed}.
\end{cases} 
\end{aligned} 
\end{equation}

\subsection{The necessary conditions for optimality}
Before proceeding with the derivation of the necessary conditions for optimality, recollect that the \textit{Link} 1 and the \textit{Link} 2 are only allowed to move such that their respective link angles lie in $]-\pi,\pi[$, and therefore, in the subsequent discussions, the base manifold $M$ is identified with $\R^2.$ 
The necessary optimality conditions for the isoholonomic problem \eqref{eq:dsopt_fixedendpoint} for the planar Purcell's swimmer is given by the following theorem: 

\begin{theorem}
Let $\{ \op{u}_k \}_{k=0}^{N-1}$ be an optimal controller that solves the problem \eqref{eq:dsopt_fixedendpoint} with $\{(\op{\alpha}_k,\op{g}_k)\}_{k=0}^{N}$ the corresponding state trajectory. There exist an adjoint trajectory $\{(\zeta^k,\xi^k)\}_{k=0}^{N-1} \subset \mathfrak{g}^*\times (\R^2)^*$ and a scalar $\nu \in \{-1,0\}$ such that, with 
\[
\rho^k \Let \Big(D \e^{-1}(\op{g}^{-1}_k\op{g}_{k+1})\circ T_e\Phi_{\op{g}^{-1}_k\op{g}_{k+1}}\Big)^{*} (\zeta^k),
\]
the following hold: 
\begin{enumerate}[label=(P-\alph*),align=left,leftmargin=*]
\item \label{con_dyn}state dynamics \eqref{kinematic},
\item \label{con_adj} adjoint dynamics  
\begin{align}\label{adjoint_dynamics}
\rho^{k-1} &= \text{Ad}^*_{\e\big(-A(\op{\alpha}_k)\op{u}_k\big)}\rho^k,\\
\xi^{k-1} & = \xi^k - h \frac{\partial f}{\partial a}(\op{\alpha}_k),
\end{align}
where 
\begin{align*}
\R^2 \ni a \mapsto f(a) \Let \big\langle \zeta^k , A(a)\op{u}_k\big\rangle \in \R, 
\end{align*}
\item \label{con_optcon}optimal control, for $\nu=-1,$
\begin{align}\label{optimal_control}
\op{u}_k = \xi^k - A^*(\op{\alpha}_k) (\zeta^k), 
\end{align}
\item \label{con_nontriv}non-triviality conditions $\{(\zeta^k,\xi^k)\}_{k=0}^{N-1}$ and $\nu$ do not vanish simultaneously.
\end{enumerate}
\end{theorem}

\begin{proof}
We now apply the discrete-time PMP presented in the Appendix A to the optimal control problem \eqref{eq:dsopt_fixedendpoint}. We define the Hamiltonian
\begin{align}\label{Hamiltonian}
& \mathfrak{g}^*\times(\R^2)^*\times G \times \R^2 \times \R^2 \ni (\zeta,\xi,g,\alpha,u) \mapsto \nonumber \\
& H^{\nu} (\zeta,\xi,g,\alpha,u) \Let \frac{h \nu}{2}\langle u, u \rangle - h \langle \zeta, A(\alpha)u\rangle + \nonumber \\
& \quad\quad\quad\quad\quad\quad\quad\quad \langle \xi, \alpha + h u \rangle \in \R
\end{align} 
for $\nu \in \R.$
Applying Theorem \eqref{thm:PMP} leads to the following set of conditions:
\begin{enumerate}
\item Condition \ref{con_dyn} is identical to condition \ref{M:sys}.
\item With the definition of the Hamiltonian \eqref{Hamiltonian}, condition \ref{M:adj} immediately leads to condition \ref{con_adj}. 
\item 
Note that $u_k \in \R^2 $ for each $k \in [N-1]$ and therefore, the Hamiltonian non-positive gradient condition leads to
\[
D_{u} H^{\nu} (\zeta^k,\xi^k,g_k,\alpha_k,u_k)= 0 \quad \text{ for each } \; k \in [N-1],
\] 
and for $\nu=-1$, it gives optimal control \ref{con_optcon}.
\item Condition \ref{con_nontriv} is identical to condition \ref{M:nontriv}.
\end{enumerate}
\end{proof}

\section{Numerical experiments}\label{sec:simulation}
We now present the numerical experiments for solving the necessary optimality conditions for an isoholonomic problem of the planar Purcell's swimmer. The viscous drag coefficient $k$ is chosen as one and also all the $3$ links are taken to be of unit length in the discrete kinematic model \eqref{dis_kin}. The time discretization step $h$ is taken as $0.01$.  Starting from the configuration of the swimmer where the base link is aligned with the x-axis of the reference frame and the $2$ outer links are extended outwards to be collinear with the base link, a translational displacement of the base link of $0.1m$ is commanded in both $x$ and $y$ directions along with the condition that the shape should return back the initial position.

	The initial conditions at $t_0=0$ seconds are taken to be:
\[
\alpha(t_0) = \big(\alpha_1(t_0), \alpha_2(t_0) \big)= (0,0),\quad g(t_0)= (0, 0, 0 ),
\] 
where the group position $g(t)$ is written as $(x(t),y(t),\theta(t))$ as shown in figure \ref{Purcell_swimmer}. The terminal conditions at $t_N=100$ seconds are:
\begin{align*}
\alpha(t_N) = \big( \alpha_1(t_N), \alpha_2(t_N) \big)=  (0,0), 
\quad g(t_N) = (0.1, 0.1, 0 ).
\end{align*}
The necessary optimality conditions \ref{con_dyn}, \ref{con_adj} with these boundary conditions form a two-point boundary value problem in $\alpha, g, \rho$ and $\xi$ with the control at each time instant is given by the condition \ref{con_optcon}. The boundary value problem is solved using Casadi \cite{Andersson2013b}, an open-source software tool with interior point technique for solving optimization problems. Figures \ref{Limb_angles_variation} to \ref{Rotational_position_variation} show the time variation of the states and control action. Figure \ref{Limb_angles_shape_space} gives the loop in the base manifold. Figure \ref{PhasePortrait} shows the time-lapse representation of the swimmer's shape and macro position at 5 second interval. The swimmer's link lengths in figure \ref{PhasePortrait} are scaled down by a factor of $50$ for better representation.

We observe that the swimmer, instead of directly moving towards the desired final position of $g(t_N) = (0.1, 0.1, 0)$, moves first in the negative $y$ direction. Also even the $x$ position of the swimmer overshoots the desired $x$ value. This is due to the constraint that there should not be any net shape change at the end of the entire maneuver. The optimal solution also results in the low values of the angular rates of the swimmer's limbs.

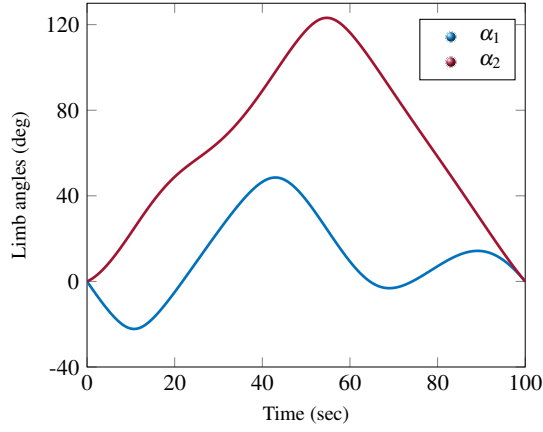
\begin{figure}
\centering
\begin{tikzpicture}[scale=0.85]
        \begin{axis}[yscale=1,xmin=0,xmax=100,ymin=-40,ymax=130, xlabel=Time (sec), ylabel near ticks, ylabel= Limb angles (deg),xtick={0,20,...,100}, ytick={-40,0,...,140}, xticklabels={0,20,...,100}, yticklabels={-40,0,...,140},label style={font=\small},]
            \addplot[skyblue, very thick] table[x=Time, y=alpha1, col sep=comma]{Data.txt};
            \addplot[redfaded, very thick] table[x=Time, y=alpha2, col sep=comma]{Data.txt};
        \end{axis}
	\begin{customlegend}[
legend entries={ 
$\alpha_1$,
$\alpha_2$
},
legend style={at={(6.7,5.45)},font=\footnotesize}] 
    \addlegendimage{mark=ball,ball color=skyblue, draw=white}
    \addlegendimage{mark=ball,ball color=redfaded,draw=white}
\end{customlegend}
\end{tikzpicture}
\caption{Limb angles trajectory}
\label{Limb_angles_variation}
\end{figure}

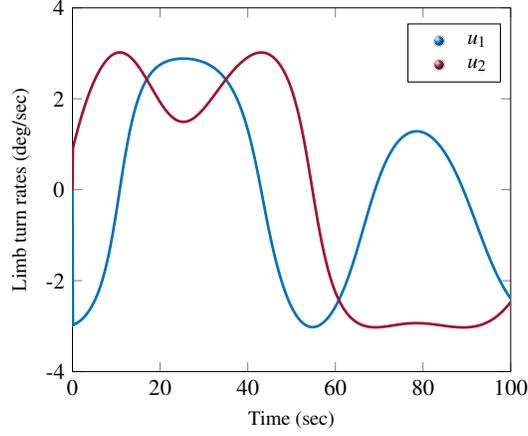
\begin{figure}
\centering
\begin{tikzpicture}[scale=0.85]
        \begin{axis}[yscale=1,xmin=0,xmax=100,ymin=-4,ymax=4, xlabel=Time (sec), ylabel near ticks, ylabel= Limb turn rates (deg/sec) ,xtick={0,20,...,100}, ytick={-4,-2,...,4}, xticklabels={0,20,...,100}, yticklabels={-4,-2,...,4},label style={font=\small},]
            \addplot[skyblue, very thick] table[x=Time, y=u1, col sep=comma]{Data.txt};
            \addplot[redfaded, very thick] table[x=Time, y=u2, col sep=comma]{Data.txt};
        \end{axis}
	\begin{customlegend}[
legend entries={ 
$u_1$,
$u_2$
},
legend style={at={(6.7,5.45)},font=\footnotesize}] 
    \addlegendimage{mark=ball,ball color=skyblue, draw=white}
    \addlegendimage{mark=ball,ball color=redfaded,draw=white}
\end{customlegend}
\end{tikzpicture}
\caption{Optimal control trajectory}
\label{Limb_rates_variation}
\end{figure}

\begin{figure}
\centering
\begin{tikzpicture}[scale=0.85]
        \begin{axis}[yscale=1,xmin=0,xmax=100,ymin=-0.35,ymax=0.35, xlabel=Time (sec), ylabel near ticks, ylabel= Position (m) ,xtick={0,20,...,100}, ytick={-0.35,-0.2,0,0.2,0.35}, xticklabels={0,20,...,100}, yticklabels={-0.35,-0.2,0,0.2,0.35},label style={font=\small},]
            \addplot[skyblue, very thick] table[x=Time, y=x, col sep=comma]{Data.txt};
            \addplot[redfaded, very thick] table[x=Time, y=y, col sep=comma]{Data.txt};
        \end{axis}
	\begin{customlegend}[
legend entries={ 
$x$,
$y$
},
legend style={at={(6.7,5.45)},font=\footnotesize}] 
    \addlegendimage{mark=ball,ball color=skyblue, draw=white}
    \addlegendimage{mark=ball,ball color=redfaded,draw=white}
\end{customlegend}
\end{tikzpicture}
\caption{Translational trajectory of the system}
\label{Translational_position_variation}
\end{figure}
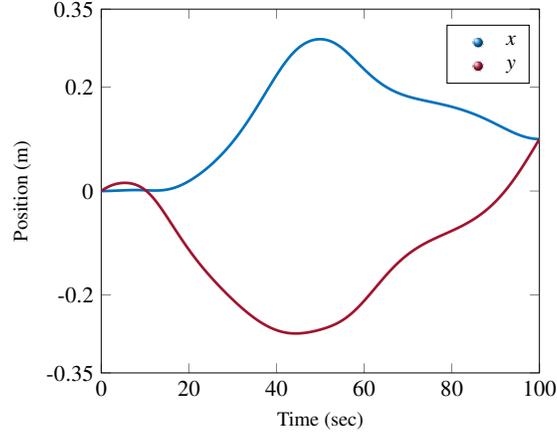

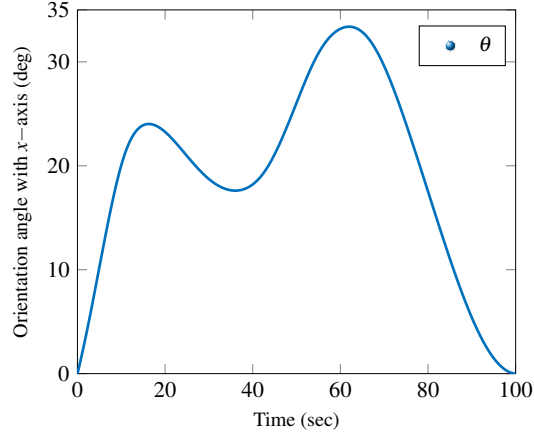
\begin{figure}
\centering
\begin{tikzpicture}[scale=0.85]
        \begin{axis}[yscale=1,xmin=0,xmax=100,ymin=0,ymax=35, xlabel=Time (sec), ylabel near ticks, ylabel= Orientation angle with $x-$axis (deg) ,xtick={0,20,...,100}, ytick={0,10,...,30,35}, xticklabels={0,20,...,100}, yticklabels={0,10,...,30,35},label style={font=\small},]
            \addplot[skyblue, very thick] table[x=Time, y=theta, col sep=comma]{Data.txt};            
        \end{axis}
	\begin{customlegend}[
legend entries={ 
$\theta$
},
legend style={at={(6.7,5.45)},font=\footnotesize}] 
    \addlegendimage{mark=ball,ball color=skyblue, draw=white}
    \addlegendimage{mark=ball,ball color=redfaded,draw=white}
\end{customlegend}
\end{tikzpicture}
\caption{Orientation trajectory on the $X-Y$ plane}
\label{Rotational_position_variation}
\end{figure}

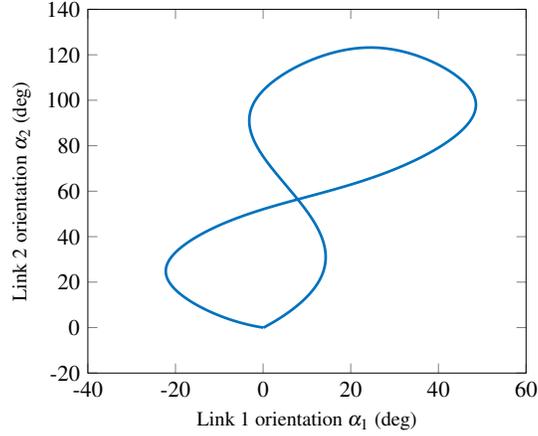
\begin{figure}
\centering
\begin{tikzpicture}[scale=0.85]
        \begin{axis}[yscale=1,xmin=-40,xmax=60,ymin=-20,ymax=140, xlabel= Link 1 orientation $\alpha_1$ (deg), ylabel= Link 2 orientation $\alpha_2$ (deg), xtick={-40,-20,...,60}, ytick={-20,0,...,140}, xticklabels={-40,-20,...,60}, yticklabels={-20,0,...,140},label style={font=\small},]
            \addplot[skyblue, very thick] table[x=alpha1, y=alpha2, col sep=comma]{Data.txt};
        \end{axis}
\end{tikzpicture}
\caption{Shape loop in the base space}
\label{Limb_angles_shape_space}
\end{figure}

\begin{filecontents*}{PhasePortrait.txt}
Time,alpha1,alpha2,x,y,theta
0,0,0,0,0,0
5,-13.886602897163902,8.657872564596605,0.001345305246427131,0.01566489242972147,10.224081849859344
10,-22.045465227994683,22.58441316080269,0.001551262812037651,0.00214081914868684,20.04025500986423
15,-17.358669426526806,37.22940329196929,0.002788071690471616,-0.04974939959558829,23.902241196712158
20,-5.001104378839936,48.84454486178105,0.01905056561699359,-0.112748388667206,23.258881133441413
25,9.211113734523662,57.08007714581396,0.05009857541811599,-0.1653203729911122,20.941811150310198
30,23.50602550530609,65.04892898284378,0.09457864622038334,-0.2088353091638056,18.72611478710488
35,36.73317258250138,75.58878919086949,0.1545819879372133,-0.2445658436430601,17.644135490180364
40,46.49840690344944,89.07993101946344,0.2229123680534603,-0.2676803023874721,18.20639774172306
45,47.551568097073144,104.03323076066694,0.2753095760832929,-0.2740049592570264,21.103559803623394
50,38.05432743448906,117.39839401167265,0.2922367305288452,-0.2667018554285684,26.05092356508969
55,23.613271391877237,123.15212156492036,0.2747577718171544,-0.2489182825511919,30.82584556279181
60,9.31660749176113,116.29290687708239,0.2343749662988845,-0.2112572884282055,33.175186269930364
65,-0.6362923121560276,102.94541176181848,0.2005047392178971,-0.1621638184579068,32.86833676092562
70,-3.022033522990844,87.94569138963598,0.1820528341048784,-0.1218353502218008,29.590633821646534
75,0.8559907374904789,72.98840847815978,0.1719214637458837,-0.09592286227383001,24.022366969354394
80,7.09333865060626,58.292084353698286,0.1621348548923953,-0.07659381448868076,17.53398700813811
85,12.487454102861589,43.47817351186553,0.1483738403668229,-0.05426316762333153,11.049707267185257
90,14.181709583412161,28.391580758883364,0.1292752834870489,-0.0200991236307959,5.395738055792195
95,9.955979298256231,13.491641060262712,0.1092075931566683,0.03241071488308356,1.489593950195526
100,0.05729577951308232,-0.05729577951308232,0.1,0.1,0
\end{filecontents*}
\DTLloaddb[noheader=false]{PhasePortraitData}{PhasePortrait.txt}
\begin{figure}
\centering
\begin{tikzpicture}[scale=0.97]
\begin{axis}[yscale=0.9,xscale=0.9,xmin=-0.05,xmax=0.35,ymin=-0.35,ymax=0.15, xlabel=$x$ (m), ylabel near ticks, ylabel= $y$ (m) ,xtick={-0.05,0.05,0.15,0.25,0.35}, ytick={-0.35,-0.25,-0.15,-0.05,0.05,0.15}, xticklabels={-0.05,0.05,0.15,0.25,0.35}, yticklabels={-0.35,-0.25,-0.15,-0.05,0.05,0.15},label style={font=\small},]
\pgfplotsextra{
\DTLforeach*{PhasePortraitData}{\Time=Time,\alphaR=alpha1, \alphaL=alpha2, \x=x, \y=y,\theta=theta}{%
\pgfmathsetmacro{\l}{0.02}
\pgfmathsetmacro{\lCxR}{\x+\l/2*cos(\theta)}
\pgfmathsetmacro{\lCyR}{\y+\l/2*sin(\theta)}
\pgfmathsetmacro{\lCxL}{\x-\l/2*cos(\theta)}
\pgfmathsetmacro{\lCyL}{\y-\l/2*sin(\theta)}
\pgfmathsetmacro{\lRxR}{\lCxR+\l*cos(\theta+\alphaR)}
\pgfmathsetmacro{\lRyR}{\lCyR+\l*sin(\theta+\alphaR)}
\pgfmathsetmacro{\lLxL}{\lCxL-\l*cos(\theta+\alphaL)}
\pgfmathsetmacro{\lLyL}{\lCyL-\l*sin(\theta+\alphaL)}
\coordinate (l0_L) at (\lCxL,\lCyL);
\coordinate (l0_R) at (\lCxR,\lCyR);
\coordinate (l1_R) at (\lRxR,\lRyR);
\coordinate (l2_L) at (\lLxL,\lLyL);

\draw[redfaded, very thick](l0_L) -- (l0_R);
\draw[lightgreen, very thick](l0_R) -- (l1_R);
\draw[skyblue, very thick](l0_L) -- (l2_L);
}
}
\draw[->, thick](0.025,-0.22)--(0.07,-0.265);
\draw[->, thick](0.255,-0.15)--(0.225,-0.08);
\end{axis}
\node [text=red] at (1.0, 3.4) {\small Start};
\node [text=red] at (2.3, 4.9) {\small End};
\end{tikzpicture}
\caption{Phase portrait of the swimmer on the $x-y$ plane}
\label{PhasePortrait}
\end{figure}
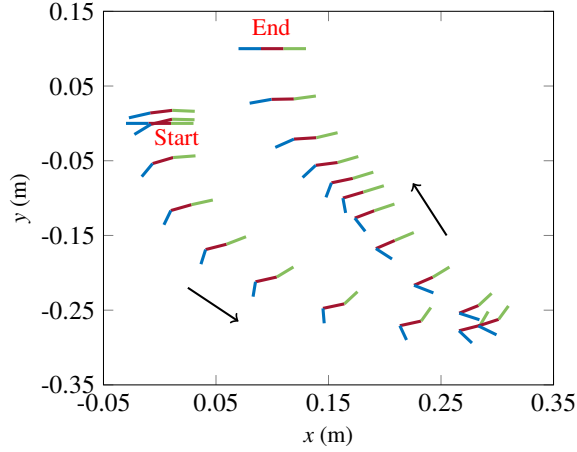



\section*{APPENDIX}

\subsection{Discrete maximum principle on matrix Lie groups}
In this section we state a modified version of a theorem on the discrete maximum principle on matrix Lie groups for fixed end point optimal control problems for systems whose kinematics evolves on a matrix Lie group $G$, and the dynamics evolves on a Euclidean space $\mathbb{R}^{n_x}$ \cite{KarmAutomatica}. The modification of the original theorem is inspired by the need of applying the discrete maximum principle to principal kinematic form of systems such as that of the Purcell's swimmer. The system model \eqref{dis_kin} is at the kinematic level and has base kinematics which evolves on $\mathbb{R}^n$ and the group kinematics which evolves on a matrix Lie group $G$. The general form of discrete-time kinematics for $t = 0, 1, \dots , N-1$ for these systems is given as follows:
\begin{equation}\label{eq:appendix_dynamics}
\begin{aligned}
x_{t+1} &= f_t(q_t,x_t,u_t), \\
q_{t+1} &= q_t s_t(q_t,x_t,u_t)
\end{aligned}
\end{equation}
where, 
\begin{align*}
&x_t \in M \subset \mathbb{R}^{n_x}, q_t \in G \text{ are the states of the system}, \\
& u_t \in \mathbb{U}_t \subset \mathbb{R}^{n_u}, \mathbb{U}_t \neq \emptyset \text{ is a set of feasible control actions}, \\
& f_t \in G \times \mathbb{R}^{n_x} \times \mathbb{R}^{n_u} \mapsto \mathbb{R}^{n_x} \text{ is a map for kinematics on } M, \\
& s_t \in G \times \mathbb{R}^{n_x} \times \mathbb{R}^{n_u} \mapsto G \text{ is a map for kinematics on $G$}. \\
& 
\end{align*}
The following assumptions are made:
\begin{enumerate}[label=(A-\alph*),align=left,leftmargin=*]
\item \label{Assumption1} The maps $s_t, f_t, g_t, c_t$ are smooth.
\item \label{Assumption2} There exists an open set $\mathcal{O} \subset G$ such that:
\begin{itemize}
\item \label{Assumption3} the exponential map $exp : \mathcal{O} \mapsto e(\mathcal{O}) \subset G$ is a diffeomorphism, and
\item \label{Assumption4} the integration step $s_t \in e(\mathcal{O}) \: \forall t$.
\end{itemize}
\item \label{Assumption5} The set of feasible control actions $\mathbb{U}_t$ is convex and compact for each $t = 0, \dots , N - 1.$
\end{enumerate}
\begin{remark}
The reader may note that there are following two key differences between the system \eqref{eq:appendix_dynamics} and that in the original theorem \cite[Theorem 2.5]{KarmAutomatica}:
\begin{enumerate}
    \item The map $s_t$ in \eqref{eq:appendix_dynamics} is a function of the control input $u_t$ in addition to $q_t, x_t$, and
    \item The dynamics in the system in \cite{KarmAutomatica} is replaced by the kinematics on the base space $M$ which has the similar evolution on $\mathbb{R}^{n_x}$ under the map $f_t$.
\end{enumerate}
However, the Pontryagin maximum principle \cite[Theorem 2.5]{KarmAutomatica} carries over with minor and obvious changes to the one that we provide below; in particular, the proof technique stays identical to that in \cite{KarmAutomatica}. We suppress the tedious details in the interest of brevity.
\end{remark}
With $c_t : G \times \mathbb{R}^{n_x} \times \mathbb{R}^{n_u} \mapsto \mathbb{R}$ as a map that accounts for the cost-per-stage for each $t=0, \dots, N-1$ the problem is to synthesize a controller by minimizing the performance index
\begin{equation}
\label{eq:appsopt}
\begin{aligned}
& \mathcal{J} (q,x,u) = \sum_{t=0}^{N-1} c_t(q_t, x_t, u_t) \\
\text{subject to} &
\begin{cases}
\text{System dynamics } \eqref{eq:appendix_dynamics},\\
u_t \in \mathbb{U}_t,\\
(q_0,x_0) = (\bar{q}_0, \bar{x}_0),\\
(q_N,x_N) = (\bar{q}_N, \bar{x}_N),\\
\text{Under the assumptions \ref{Assumption1} to \ref{Assumption5}}.
\end{cases} 
\end{aligned} 
\end{equation}

\begin{theorem}\label{thm:PMP}
Let $\{ u^o_t \}^{N-1}_{t=0}$ be an optimal controller that solves the above problem. Define a Hamiltonian 
\begin{align*}
&[N] \times \mathfrak{g}^{*} \times (\mathbb{R}^{n_x})^{*} \times G \times \mathbb{R}^{n_x} \times \mathbb{R}^{n_u} \ni (\tau , \zeta , \xi , q, x ,u) \mapsto \\
& H^{\nu} (\tau, , \zeta , \xi , q, x ,u) = \nu c_{\tau} (q,x,u) + \langle \zeta , \e^{-1}(s_{\tau}(q_{\tau},x,u)) \rangle_{\mathfrak{g}} \\
& \qquad \qquad \qquad \qquad + \langle \xi , f_{\tau} (q,x,u)\rangle \in \mathbb{R}
\end{align*}
for $\nu \in \mathbb{R}$. Then there exists 
\begin{itemize}
    \item an adjoint trajectory $\{ (\zeta^t, \xi^t)\}_{t=0}^{N-1} \subset \mathfrak{g}^* \times (\mathbb{R}^{n_x})^{*}$, and
    \item a scalar $\nu \in \{-1, 0\}$
\end{itemize}
such that, with 
\begin{equation}\nonumber
\gamma_t = (t, \zeta^t, \xi^t, \mu^t, q_t, x_t, u_t),
\end{equation}
the following conditions hold:
\begin{enumerate}[label=(M-\alph*),align=left,leftmargin=*]
\item \label{M:sys} state dynamics
\begin{align*}
x_{t+1} &= D_{\xi} H^{\nu} (\gamma_t),  \\ 
q_{t+1} &=q_t \e^{D_{\zeta} H^{\nu} (\gamma_t)},
\end{align*}
\item \label{M:adj} adjoint dynamics
\begin{align*}
\xi_{t-1} &= D_x H^{\nu} (\gamma_t),\\
\zeta^{t+1} &= Ad^{*}_{\e^{-D_{\zeta} H^{\nu}, (\gamma_t)}}\zeta^t + T_e^{*}\Phi_{q_t} (D_q H^{\nu} (\gamma_t),
\end{align*}
\item \label{M:ham_max} Hamiltonian non-positive gradient condition
\begin{equation}\nonumber
\langle D_u H^{\nu} (t, \zeta^t, \xi^t, q_t, x_t, u_t), w - \mathring{u}_t \rangle \leq 0  \text{ for all } w \in \mathbb{U}_t,
\end{equation}
\item \label{M:nontriv} Non-triviality condition: The adjoint variables $\{ (\zeta^t, \xi^t) \}_{t=0}^{N}$ and the scalar $\nu$ do not vanish simultaneously.
\end{enumerate}
\end{theorem}

\bibliographystyle{ieeetr}
\bibliography{SK_RNB_KSP_DC_CDC2018}
\end{document}